\documentclass[runningheads]{llncs}
\usepackage[T1]{fontenc}
\usepackage{graphicx}
\usepackage{multirow}

\usepackage{esvect}
\usepackage{amsmath,amssymb,amsfonts}
\usepackage{amsthm}
\usepackage{mathrsfs}
\usepackage[title]{appendix}
\usepackage{xcolor}
\usepackage{hyperref}
\usepackage{textcomp}
\usepackage{manyfoot}
\usepackage{booktabs}
\usepackage{algorithm}
\usepackage{algorithmicx}
\usepackage{algpseudocode}
\usepackage{listings}
\usepackage{tikz}
\usetikzlibrary{positioning, arrows.meta}
\usepackage{fullpage}

\begin{document}
\title{Hierarchical Identity-Based Signature with Designated Aggregator from Lattices}
\author{
Stuti Kumari\inst{1} \and
Kunal Dey\inst{2} \and 
Vikas Srivastava\inst{3} \and
Sumit Kumar Debnath\inst{1}
}

\institute{
$^{1}$National Institute of Technology Jamshedpur, Jamshedpur 831014, India\\
\email{2023rsma004@nitjsr.ac.in}, 
\email{sdebnath.math@nitjsr.ac.in}\\
$^{2}$ SRM University-AP, Andhra Pradesh, India-522240\\
\email{kunaldey3@gmail.com},
\email{kunal.d@srmap.edu.in}\\
$^{3}$ National Institute of Technology, Warangal - 506004, Telangana, India\\
\email{vikas.math123@gmail.com}
}
 \maketitle   

\begin{abstract} 
In hierarchical organizations, authenticating data from multiple users can be complex and resource-intensive. Hierarchical Identity-Based Signature with Designated Aggregator (HIBS-DA) provides an efficient solution by allowing users at different levels to generate signatures that can be combined into a single, compact signature.  We first introduce the HIBS-DA framework and present the {\em{first}} lattice-based construction of HIBS-DA.
Our scheme allows users at different hierarchical levels to generate individual signatures that can be aggregated into a single, compact signature, reducing communication and verification costs. The proposed construction is secure, correct, and resistant to forgery, making it suitable for large-scale environments such as universities, corporations, and government agencies.
\end{abstract}


\section{Introduction}
In classical public-key cryptography, a certification authority (CA) issues digital certificates to users to bind their identities to public keys, resulting in significant certificate management overhead. The idea of identity-based cryptography was originally put forward by Shamir (1985) \cite{shamir1984identity}, wherein identities like email addresses function as public keys, with corresponding secret keys issued by a PKG. Boneh and Franklin \cite{boneh2003identity} introduced the earliest practical scheme for Identity-based encryption in 2001, based upon the bilinear Diffie-Hellman (BDH) assumption, which subsequently led to further constructions including IBS schemes \cite{hess2002efficient}.
Subsequently, various extensions have appeared in the literature, including identity-based signature (IBS) schemes \cite{hess2002efficient}. However, IBS typically employs a single PKG, which can become a performance bottleneck, making such systems unsuitable for large organizations due to the potential overload on the PKG.\\
In many modern applications, such as secure data sharing \cite{wang2018lattice}, e-governance \cite{kansal2022efficient}, and distributed systems \cite{alkadri2024practical}, it is often necessary to authenticate data from multiple users organized in hierarchical structures. Moreover, when numerous signatures are involved, storing and verifying each individually becomes costly in terms of time and space. To address these challenges, Hierarchical Identity-Based Signature with Designated Aggregator (HIBS-DA) schemes have been introduced. An HIBS-DA scheme combines the advantages of hierarchical identity-based cryptography (HIBC) \cite{gentry2002hierarchical} and aggregate signatures \cite{boneh2003aggregate}.
In a hierarchical identity-based setting, users are organized into a tree-like structure with a trusted Private Key Generator (PKG) at the root. The PKG delegates key generation down the hierarchy, enabling scalable and decentralized key management. This is particularly useful in large organizations, government departments, academics, hospitals, companies where authority and access are distributed across multiple levels.\\
On the other hand, an aggregate signature scheme \cite{boneh2003aggregate} allows multiple signatures possibly from different users and on different messages, to be merged into one compact signature. This aggregated signature can subsequently be verified efficiently, authenticating all underlying messages simultaneously.\\
An HIBS-DA scheme, therefore, supports
signatures by multiple users at different levels of a hierarchy,
Identity-based key generation without relying on certificates.
Such a system is highly beneficial in scenarios where multiple users from different branches of a hierarchy need to sign data, and the verifier requires a compact and efficient authentication mechanism. \\
In this setting, a single private key generator (PKG) and multiple hierarchical groups of users exist. Each group forms a hierarchy of depth $t$, with a total of $N$ such hierarchies, each having the same depth $t$.
Now, suppose for a specific application, it is required to obtain documents signed by the $k$-th level signer from each hierarchy. That is, the signers denoted as $s_{1k}, s_{2k}, \ldots, s_{Nk}$ from the figure \ref{fig:HIBS-DA} corresponding to the $k$-th level user in each of the $N$ hierarchies are required to generate signatures on their respective data. These signatures, along with the signed data, must be submitted wherever needed and stored for verification.
Instead of storing and verifying $N$ individual signatures, we can leverage an aggregate signature scheme. This allows all the individual signatures from the $k$-th signer in each hierarchy to be compressed into a single, compact aggregate signature. This one signature is sufficient to authenticate all the corresponding documents, significantly improving efficiency in terms of storage and verification. In this paper, we begin the investigation of HIBS-DA. We formalize its syntax, security architecture, and provide the {\em{first}} lattice-based construction, and evaluate its efficiency using comprehensive complexity estimates.

\subsection{Application}
The hierarchical identity-based signature with designated aggregator (HIBS-DA) schemes have several practical applications, we present an application of HIBS-DA in detail as follows.

Consider an academic institution organized in a hierarchical structure comprising a Director at the top level, followed by $N$ Heads of Departments (HoDs) ($h_1, h_2, \ldots, h_N$), and $N$ Financial Officers ($f_1, f_2, \ldots, f_N$), with authority delegated in that order. This hierarchical structure reflects a natural delegation model, where the Director acts as the root Private Key Generator (PKG), enabling identity-based key generation throughout the levels of the hierarchy.\\
Now, consider two scenarios for a particular academic year:
\begin{enumerate}
\item[(i)] Departmental Placement Records Authentication:
Each HoD ($h_i$) is responsible for submitting the placement records of their respective departments, digitally signed to ensure authenticity and integrity. Rather than storing and verifying all $N$ individual signatures, a designated aggregator can be employed to compress these signatures into a single aggregate signature. This compact signature can then be used to authenticate all placement records of different departments at once, improving efficiency in verification and storage.
\item[(ii)]Departmental Financial Records Authentication:
Each Financial Officer ($f_i$) submits the financial report for their respective department, again digitally signed. As in the previous case, the individual signatures from all financial officers can be aggregated into a single signature. This single compact signature serves as a proof of authenticity for all departmental financial records.
\end{enumerate}
\begin{figure}[ht!]
   \centering
 \includegraphics[width=0.9\textwidth]{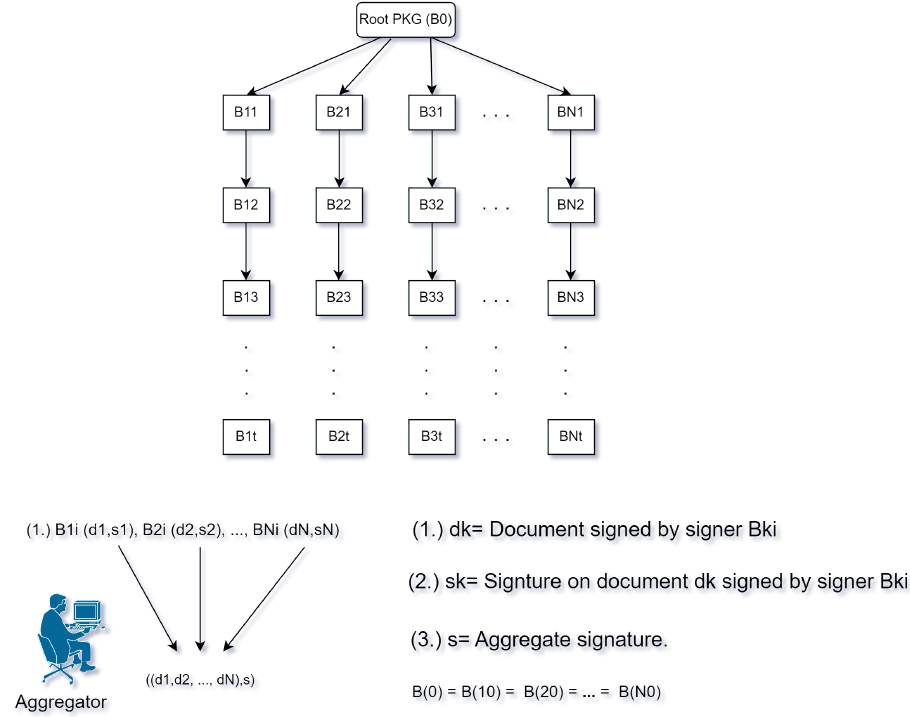}
   \caption{HIBS-DA.}
   \label{fig:HIBS-DA}
 \end{figure}

\noindent By employing the HIBS-DA scheme in such a setting, the institution achieves:
Efficient key delegation aligned with its organizational hierarchy,
Compact and scalable authentication of multiple signed documents,
Post-quantum security is built upon lattice-based assumptions.
Such a model is ideal for hierarchical organizations that handle structured, multi-source data submissions while aiming to minimize computational and communication overhead. Refer to Figure~\ref{fig:HIBS-DA} for a visual demonstration.


\subsection{Review of Literature}
An overview of current lattice-based approaches to hierarchical identity-based signatures and aggregate signatures is presented in this section.\\
Gentry and Silverberg \cite{gentry2002hierarchical} introduced the hierarchical identity-based signature (HIBS) scheme, which extends traditional identity-based signature (IBS) systems by incorporating a multi-level hierarchical structure, similar to those found in real-world organizations. In a HIBS scheme, multiple private key generators (PKGs) are arranged hierarchically, with each PKG's secret key being generated by its parent. This approach reduces the computational burden on the root PKG, making it highly suitable for large-scale deployments. Most HIBS schemes in the literature \cite{gentry2002hierarchical,chow2004secure,zhang2009new} rely on the difficulty of solving the discrete logarithm problem. However, in 1997, Shor demonstrated that quantum algorithms could solve this problem efficiently, posing a threat to such systems.\\
Several lattice-based HIBS schemes have been developed.
Rückert \cite{ruckert2010strongly}, inspired by the basis delegation techniques of \cite{cash2012bonsai}, introduced 
the first HIBS scheme built on lattices, both in the random-oracle setting and in the standard model. Despite these advances, these schemes tend to involve large system parameters, and the efficiency of both the private key generation and the signing process depends on the signer's position within the hierarchy. To enhance efficiency,
Tian et al. \cite{tian2012new} introduced a novel HIBS scheme based on lattices that does not rely on random oracles and offers better performance.
Then, lattice-based HIBS schemes \cite{tian2013efficient,srivastava2020hierarchical,yi2019access}
have been proposed.\\
 The idea of aggregate signatures was initially proposed by Boneh et al.\cite{boneh2003aggregate}. It merges several signatures, each linked to distinct messages, into a single compact signature. This compact signature can authenticate several message-signature pairs for different users simultaneously. Such a scheme significantly reduces the storage space needed for signatures, lowers the transmission bandwidth demand, and minimizes the computational cost of verifying signatures. 
Various lattice-based aggregate signature schemes \cite{shen2016provably,el2014towards} have been developed.\\

\section{Our Contributions}
\label{contri}
To the best of our knowledge, no HIBS-DA scheme has been proposed in the existing literature. Therefore, we propose to construct a lattice-based HIBS-DA scheme to address this gap.
We propose the {\em{first}} HIBS-DA scheme based on lattices that simultaneously supports hierarchical key delegation and aggregation of signatures from multiple branches in a hierarchical PKG structure. 

\noindent To achieve the HIBS-DA scheme, we integrate the Chinese Remainder Theorem (CRT) technique described in (see in Appendix \ref{crt},\ref{INT}) with the hierarchical identity-based signature scheme proposed by Tian et al.~\cite{tian2012new}. The construction of our HIBS-DA scheme is presented in Section \ref{CONS}, while its correctness and security proofs are provided in the subsequent subsections. The detailed efficiency analysis of our HIBS-DA scheme, describing round complexity, communication complexity, and computation complexity, is given in Section \ref{eff ana}.\\

\noindent In this part, we give a quick summary of our HIBS-DA framework. Section \ref{HI} provides a formal syntax, a full security analysis as well as a complete discussion of the methods.

\noindent \textbf{Syntax.} A HIBS-DA scheme is defined by six polynomial time algorithms
{HIBS-DA=(\sf HIBS-DA.Setup, \sf HIBS-DA.Derive, \sf HIBS-DA.Sign, \sf HIBS-DA.Verify, \sf HIBS-DA.Aggregate, \sf HIBS-DA.AggregateVerify)}.
\begin{itemize}
\item[$\bullet$] {\sf HIBS-DA.Setup}($1^{\lambda},1^{t}) \rightarrow (pp, MPK, MSK)$: Takes security parameter $\lambda$ and maximum hierarchy depth $t$; outputs public parameters $pp$, master public key $MPK$, and master secret key $MSK$.

\item[$\bullet$] {\sf HIBS-DA.Derive}($pp, SK_{ID|_{ir}}, ID|_{ik}) \rightarrow SK_{ID|_{ik}}$: Derives child identity secret key $SK_{ID_{|ik}}$ from parent key $SK_{ID_{|ir}}$.

\item[$\bullet$] {\sf HIBS-DA.Sign}$(ID|_{ik},M_{ik}, \mathsf{SK}_{ID|_{ik}}) \rightarrow v_{ik}$: Produces signature $v_{ik}$ on message $M_{ik}$ for identity $ID_{|ik}$.

\item[$\bullet$] {\sf HIBS-DA.Verify}($pp, {ID|_{ik}}, M_{ik}, v_{ik}) \rightarrow (1/0)$: Verifies signature $v_{ik}$ for message $M_{ik}$ under identity $ID_{|ik}$.


\item[$\bullet$] {\sf HIBS-DA.Aggregate} \big($\{M_{ik}\}_{i=1}^N ,\{ v_{ik}\}_{i=1}^N\big) \rightarrow v_k$: Combines $N$ message-signature pairs into aggregate signature $v_k$.


\item[$\bullet$] {\sf HIBS-DA.AggregateVerify}($\{M_{ik}\}_{i=1}^N, v_k, MPK) \rightarrow (1/0)$: Verifies aggregate signature $v_k$ against messages $\{M_{ik}\}_{i=1}^N$ using $MPK$.

\end{itemize}
\noindent \textbf{Security.} The standard notion for security used for our HIBS-DA scheme is Existential Unforgeability under Adaptive Identity and Chosen Message Attack (EUF-ID-CMA). In this security game, a PPT adversary $\mathcal{A}$ interacts with a challenger $\mathcal{C}$, attempting to forge a valid aggregate signature for a specific hierarchical level $k$ after making adaptive key and signature queries. The scheme is EUF-ID-CMA secure if for all PPT adversaries $\mathcal{A}$, the advantage $\text{Adv}^{\text{HIBS-DA}}_{\mathcal{A}}(\lambda)$ in winning this game is negligible in the security parameter $\lambda$.\\

\noindent \textbf{Technical Overview.} We now present a technical overview of our proposed HIBS-DA scheme. The system is built upon a cryptographic infrastructure headed by a Root Private Key Generator (Root PKG) at level $0$, which oversees a hierarchy of depth $t$ with levels indexed from $1$ to $t$. Each level $k$ (where $1 \leq k \leq t$) consists of a fixed population of $N$ entities. An entity at the $k$-th level of branch $i$ is uniquely identified by the tuple $ID|_{ik} = (ID_{i1}, \dots, ID_{ik})$. Moreover, $ID_0 = ID_{i0}$ is the root PKG’s identity for each $i = {1,2, \ldots, N}$.
 We consider a hierarchical setting in which, for a fixed hierarchical level $k \in \{1, \ldots, t\}$, each of the $N$ signers generates an individual signature on its corresponding message. A designated aggregator then collects 
all individual signatures and combines them into a single compact aggregate signature. 
The HIBS-DA framework is structured into six core phases (Setup, Derive, Sign, Verify, Aggregate, and AggregateVerify). The Setup phase, executed by a trusted third party (TTP), takes a security parameter $1^\lambda$ and maximum depth $1^t$ to produce a master public key $MPK = \mathbf{A}$ and a corresponding master secret key $MSK = \mathbf{T_A}$ for the root PKG by using the TrapGen algorithm. The TTP also provides the designated aggregator with a short basis $\mathcal{T}$ for an intersection lattice, which is defined in Appendix \ref{crt} by using Lemma \ref{homo}.
In the Derive phase, a child entity at identity $ID|_{ik}$ derives its secret key from its parent's secret key $SK_{ID|_{ir}}$. This involves computing a public matrix $P_{ID|_{ik}}$ from the parent's matrix and the child's identity component, and then using the $\textsf{BasisDel}$ algorithm to derive a basis consisting of short vectors $\mathbf{S'_{ik}}$ for the child's lattice, which becomes its secret key $SK_{ID|_{ik}}$. 
For the Sign phase, a user $ID|_{ik}$ signs a message $M_{ik}$ by generating a lattice-based signature $v_{ik}$ using the $\textsf{SamplePre}$ algorithm with their secret key, which a verifier checks for validity. Finally, in the Aggregate phase, the designated aggregator combines $N$ individual signatures $\{v_{ik}\}_{i=1}^N$ on messages $\{M_{ik}\}_{i=1}^N$ into a single compact signature $v_k$ using the CRT (refer\ref{crt},\ref{INT}) and the $\textsf{SamplePre}$ algorithm. The resulting aggregate signature $v_k$ can be efficiently verified to authenticate the entire set of messages and identities from level $k$.\\
\textbf{Security:} The proposed HIBS-DA scheme is existentially unforgeable under adaptive identity and chosen message attack (EUF-ID-CMA) under the hardness assumption of the SIS problem.
\begin{figure}[ht!]
   \centering
 \includegraphics[width=0.9\textwidth]{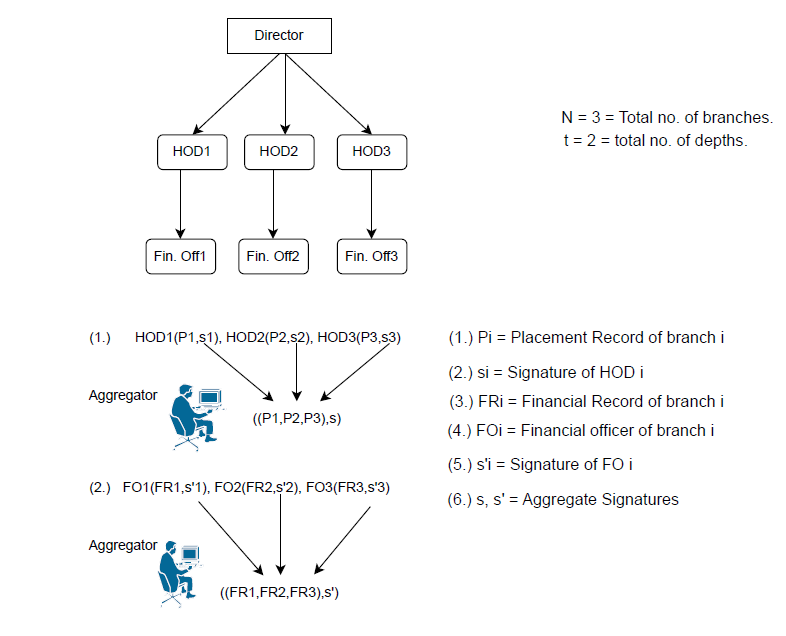}.
   \caption{Practical application of HIBS-DA.}
   \label{fig: app1}
 \end{figure}

\subsection{A Toy Example}
Figure~\ref{fig: app1} presents a toy example of our proposed HIBS-DA scheme. The scenario involves an institute with a Director, designated as the Root PKG, managing a system with three branches ($N=3$) and a hierarchy depth $t=2$.\\
The scheme operates as follows:\\

\noindent Departmental Record Aggregation: The Head of Department (HOD) at each branch, i.e., of each department signs their respective Placement Record $P_i$. Specifically, ${HOD}_i$ signs $P_i$ to produce signature $s_i$, for each $i= 1, 2, 3$. A designated aggregator then combines these three signatures $\{s_1, s_2, s_3\}$ into a single aggregate signature $s$, which validates the entire collection of placement records $(P_1, P_2, P_3)$.\\
\noindent Financial Record Aggregation:  In a parallel process, the financial officer $({FO}_i)$ of each department generates a signature  $s'_i$  on their respective Financial Record $FR_i$. A designated aggregator then compresses these signatures $\{s'_1, s'_2, s'_3\}$ into a single aggregate signature $s'$. This signature $s'$ provides collective authentication for the complete set of financial records $({FR}_1, {FR}_2, {FR}_3)$.\\

\noindent This approach allows the Director or any designated party to verify all placement records and financial records of each department using only two aggregate signatures ($s$ and $s'$) instead of six individual ones, demonstrating significant efficiency gains.

\section{Prelminaries}
\textbf{Notation.} For an integer $ q\geq 2$, let  $\mathbb{Z}_q$ represents the ring of integers modulo $q$, and $\mathbb{Z}_q^{n\times m}$ denotes the collection of all $n\times m$ matrices whose entries belong to $\mathbb{Z}_q$. Throughout, vectors are expressed using bold lowercase letters, while matrices are indicated by bold uppercase letters.
\begin{definition} \textbf{(Lattices)~\cite{micciancio2002complexity}:}
A lattice in the Euclidean space $\mathbb{R}^n$ is generated by a set of $m$ linearly independent vectors $\mathbf{y}_1, \ldots, \mathbf{y}_m \in \mathbb{R}^n$ with $n \ge m$. It is defined as
\[\mathbb{L}(\mathbf{y}_1, \ldots, \mathbf{y}_m)
= \left\{ \sum_{i=1}^{m} c_i \mathbf{y}_i : c_i \in \mathbb{Z} \right\}
\]
that is, the collection of all integer linear combinations of the basis vectors $\mathbf{y}_i$.
Here, $m$ is the rank, and $n$ is the dimension of the lattice.  
The set of vectors $\mathbf{y}_1, \ldots, \mathbf{y}_m$ is referred to as a lattice basis.\\
For any \( \mathbf{A} \in \mathbb{Z}_q^{n \times m}, \) we define

\[
\mathbb{L}_q^{\perp}(\mathbf{A}) = \left\{ \mathbf{y} \in \mathbb{Z}^m : \mathbf{A}\mathbf{y} = \mathbf{0} \bmod q \right\}.\]
\[
\mathbb{L}_q^{\mathbf{u}}(\mathbf{A}) = \left\{ \mathbf{y} \in \mathbb{Z}^m : \mathbf{A}\mathbf{y} = \mathbf{u} \bmod q \right\}
\]
The lattice $\mathbb{L}_q^{\mathbf{u}}(\mathbf{A})$  is a coset of $\mathbb{L}_q^{\perp}(\mathbf{A})$; that is, $\mathbb{L}_q^{\mathbf{u}}(\mathbf{A})=\mathbb{L}_q^{\perp}(\mathbf{A})+ \mathbf{t}$ for any vector $\mathbf{t}$ satisfying $\mathbf{A} . \mathbf{t} = \mathbf{u} \bmod q.$
\end{definition} 
\begin{definition}\label{trapgen}
\textbf{TrapGen($n, m, q$)~\cite{gentry2008trapdoors}:}  
The algorithm TrapGen is a probabilistic procedure that, given system parameters $n$, matrix width $m$, and a modulus $q$, outputs a pair
TrapGen$(n,m,q) \;\longrightarrow\; (\mathbf{A}, \mathbf{T})$,
where $\mathbf{A} \in \mathbb{Z}_q^{n \times m}$ is a matrix chosen  uniformly at random and $\mathbf{T}$ is a trapdoor basis having low norm for the lattice
$\Lambda_q^\perp(\mathbf{A}) = \{\, \mathbf{x} \in \mathbb{Z}^m : \mathbf{A} \mathbf{x} \equiv 0 \pmod{q} \,\}$.
\end{definition}
\begin{definition}
\textbf{SamplePre($\mathbf{T}$,$\mathbf{u}$,$t$)~\cite{gentry2008trapdoors}:}
The sampling algorithm, runs in probabilistic polynomial time (PPT), is given as input a pair $(\mathbf{A}, \mathbf{T})$ generated by the TrapGen($n, m, q$) algorithm, a vector $\mathbf{u} \in \mathbb{Z}^n$, and a Gaussian parameter $t$. It produces a vector $\mathbf{x} \in \mathbb{Z}^m$ such that $\|\mathbf{x}\| \leq t \sqrt{m}$ and satisfies the modular equation $\mathbf{A} \mathbf{x} = \mathbf{u} \bmod q$.
\end{definition}
\begin{definition}
\textbf{Domain sampling and uniform output \cite{gentry2008trapdoors}}: $\text{SampleDom}(1^n)$ samples an $x$ from some (possibly non-uniform) distribution over $D_n = \{e \in \mathbb{Z}^m : \|e\| \leq s \sqrt{m}\}$, for which the distribution of $\mathbf{A}\mathbf{x}$ is uniform over $\mathbb{Z}_q^n$, where $s$ is some prescribed bound. 
\end{definition}
\begin{definition}
\textbf{preimage min-entropy \cite{gentry2008trapdoors}}: \label{meh} 
Let $\mathbf{A} \in \mathbb{Z}_q^{n \times m}$
For each $\mathbf{b} \in \mathbb{Z}_q^n$, the conditional min-entropy of $\mathbf{c}$ sampled from $\text{SampleDom}(1^n)$, given that $\mathbf{A}\mathbf{c} = \mathbf{b}$, is no less than $\omega(\log n)$.
\end{definition}

\begin{definition}
\textbf{SampleRwithBasis \cite{agrawal2010lattice}}  
Let $\mathbf{c}_1, \ldots, \mathbf{c}_m \in \mathbb{Z}_q^n$ be the column vectors comprising the matrix 
$\mathbf{C} \in \mathbb{Z}_q^{n \times m}$.

\begin{enumerate}
    \item Execute $\text{TrapGen}(q,n,m)$ to obtain a uniformly distributed full-rank matrix
    $\mathbf{D} \in \mathbb{Z}_q^{n \times m}$ together with a trapdoor basis 
    $\mathbf{T_D}$ for $\Lambda_q^\perp(\mathbf{D})$ satisfying
    \[
        \|\widetilde{\mathbf{T}_D}\| \le \widetilde{L}_{T_G} 
        = \sigma_R / \omega(\sqrt{\log m}) .
    \]

    \item For each $i = 1, \ldots, m$:
    \begin{enumerate}
        \item[(2a)] Draw 
        $\mathbf{r}_i \leftarrow \text{SamplePre}(\mathbf{D}, \mathbf{T_D}, \mathbf{c}_i, \sigma_R)$.
        Then $\mathbf{D}\mathbf{r}_i = \mathbf{c}_i \bmod q$. Moreover, the statistical distribution of 
        $\mathbf{r}_i$ is close to 
        $D^{\mathbf{c}_i}_{\Lambda_q^\perp(\mathbf{D}), \sigma_R}$.

        \item[(2b)] Perform step (2a) repeatedly until $\mathbf{r}_i$ is no longer a linear combination of
$\mathbf{r}_1, \ldots, \mathbf{r}_{i-1}$.
    \end{enumerate}

    \item Consider $\mathbf{R} \in \mathbb{Z}^{m \times m}$ constructed by placing
$\mathbf{r}_1, \ldots, \mathbf{r}_m$ as its columns. By construction, $\mathbf{R}$ has full rank modulo $q$. 
    Output $\mathbf{R}$ and $\mathbf{T_D}$.
\end{enumerate}

\noindent Since $\mathbf{D}\mathbf{R} = \mathbf{C} \bmod q$, it follows that 
$\mathbf{D} = \mathbf{C}\mathbf{R}^{-1} \bmod q$. Therefore, $\mathbf{T_D}$ is a basis consisting of short vectors for
$\Lambda_q^\perp(\mathbf{C}\mathbf{R}^{-1})$. We need to verify that $\mathbf{R}$ is drawn according to some distribution that is close to $\mathcal{D}_{m\times m}$ statistically.
\end{definition}
\begin{definition}
\textbf{Short Integer Solution(SIS)~\cite{gentry2008trapdoors}:}
Suppose $q$ be an integer, $\mathbf{A}$ be a matrix from $\mathbb{Z}_q^{n \times m}$, and a real parameter $\beta$, the target is to find a nonzero integer vector $\mathbf{x} \in \mathbb{Z}^m$ that satisfies $\mathbf{A}\mathbf{x} \equiv \mathbf{0} \bmod{q}$ and whose norm is at most $\beta$.
\end{definition}

\noindent \textbf{Distributions over small-norm matrices.}
A matrix $R \in \mathbb{Z}^{m \times m}$ has an 
inverse over $\mathbb{Z}_q$, or is called invertible when its reduction modulo $q$ remains invertible in 
$\mathbb{Z}_q^{m \times m}$. Our construction operates with matrices possessing this $\mathbb{Z}_q$-invertibility whose columns all have small Euclidean norm.

\begin{definition}
Set 
\[
\sigma_R := \widetilde{L}_T G \cdot \omega(\sqrt{\log m})
           = \sqrt{n \log q}\,\omega(\sqrt{\log m}) .
\]
\end{definition}

\noindent We define $\mathcal{D}_{m \times m}$ as the distribution over matrices in 
$\mathbb{Z}^{m \times m}$ obtained by drawing each column independently from 
$\mathcal{D}_{\mathbb{Z}^m, \sigma_R}$ and conditioning on ensuring that the resulting matrix is invertible over 
$\mathbb{Z}_q$.

\begin{definition}{\textbf{SampleR$(1^m)$}.}
The procedure \textsf{SampleR}$(1^m)$ generates matrices over $\mathbb{Z}^{m \times m}$ whose distribution is within negligible statistical distance of $\mathcal{D}_{m\times m}$.

\begin{enumerate}
    \item Consider $B$ denote the standard lattice basis of $\mathbb{Z}^m$.
    \item For each $j = 1, \ldots, m$, sample
    \[
        s_i \gets \mathsf{SampleGaussian}(\mathbb{Z}^m, B, \sigma_R, 0).
    \]
\end{enumerate}
\begin{enumerate}
\item Form the matrix $R$ from the vectors $s_i$. Output $R$ when it is invertible modulo $q$, otherwise, perform step~2 again.
\end{enumerate}
\end{definition}

\begin{definition}
{\textbf{Basis delegation: {BasisDel}$(A, R, T_A, \sigma)$.}}\label{BD}
We next introduce an algorithm for basis delegation which maintains the dimensions of the matrices throughout the computation.  
This procedure allows one to transform a given lattice basis through a $\mathbb{Z}_q$-invertible matrix while preserving the structural properties required for lattice-based cryptographic constructions.

\noindent \textbf{Inputs:}
\begin{itemize}
    \item A rank-$n$ matrix $A \in \mathbb{Z}_q^{n \times m}$, representing the original lattice constraints.
    \item A $\mathbb{Z}_q$-invertible matrix $R \in \mathbb{Z}^{m \times m}$, sampled from the distribution $\mathcal{D}_{m \times m}$ (or a product of such matrices), which defines the transformation applied to the lattice basis.
    \item A basis $T_A$ of the lattice $\Lambda_q^{\perp}(A)$.
    \item A positive real parameter $\sigma \in \mathbb{R}_{>0}$, controlling the Gaussian sampling in subsequent computations.
\end{itemize}

\noindent \textbf{Output:}  
The algorithm computes $B := A R^{-1} \in \mathbb{Z}_q^{n \times m}$ and outputs a basis $T_B$ of the lattice $\Lambda_q^{\perp}(B)$.  
This basis $T_B$ effectively delegates the structure of $T_A$ through the transformation induced by $R$.
\end{definition}
\begin{lemma}~\cite{micciancio2002complexity} 
\label{homo}Let $\mathbb{L}$ be an $n$-dimensional lattice. There is a polynomial time algorithm which, given a lattice basis $\mathcal{B}$ and linearly independent vectors $\{\mathbf{b}_1, \ldots, \mathbf{b}_m\} \subseteq \mathbb{L}(\mathcal{B})$ ordered by non-decreasing norm such $\| \mathbf{b}_1 \| \leq \| \mathbf{b}_2 \| \leq \cdots \leq \| \mathbf{b}_m \|$, outputs a basis $\mathcal{R} = \{\mathbf{r}_1, \ldots, \mathbf{r}_m\}$ equivalent to $\mathcal{B}$. The output basis satisfies $\| \mathbf{r}_k \| \leq \max \{\left( \frac{\sqrt{k}}{2}\right) \| \mathbf{b}_k \|, \| \mathbf{b}_k \| \}$, for all $k = 1,2, \ldots m$. Furthermore, it holds that $\text{span}(\mathbf{r}_1, \ldots, \mathbf{r}_k) = \text{span}(\mathbf{b}_1, \ldots, \mathbf{b}_k)$
and $\| \mathbf{r}_k^* \| \leq \| \mathbf{b}_k^* \|$ for every $k = 1,2, \ldots m.$
\end{lemma}

\begin{lemma}\cite{agrawal2010lattice} \label{bd}Let \(\mathbf{A} \in \mathbb{Z}_q^{n \times m}\) be a rank-\(n\) matrix, and let \(\mathbf{R} \in \mathbb{Z}^{m \times m}\) be a \(\mathbb{Z}_q\)-invertible matrix sampled from the distribution \(\mathcal{D}_{m \times m}\) (or obtained as a product of such matrices). Let \(\mathbf{T_A}\) be a basis of the lattice \(\Lambda_q^{\perp}(\mathbf{A})\), and let \(\sigma \in \mathbb{R}_{>0}\) be a parameter.
Then the algorithm \(\textbf{BasisDel}(\mathbf{A}, \mathbf{R}, \mathbf{T_A}, \sigma)\) outputs a basis \(\mathbf{T_B}\) of the lattice \(\Lambda_q^{\perp}(\mathbf{B})\), where $\mathbf{B}= \mathbf{A} \mathbf{R^{-1}}$.
\end{lemma}

\subsection{Chinese Remainder Theorem (CRT)~\cite{atiyah2018introduction}:}
\begin{theorem}\label{crth}
Consider a commutative ring $R$ with unity, and suppose $S_1, S_2, \ldots, S_N$ be ideals of $R$ that are pairwise comaximal, i.e.,
\[
S_i + S_j = R \quad \text{for all } i \neq j.
\]
Then there is a ring isomorphism
\[
\frac{R}{\bigcap_{i=1}^N S_i} \;\cong\; \prod_{i=1}^{N} \frac{R}{S_i},
\]
given by
\[
r \;\mapsto\; (r \bmod S_1, r \bmod S_2, \dots, r \bmod S_N), \quad r \in R.
\]
\end{theorem}

\subsection{Aggregate Signature}
An aggregate signature enables $N$ distinct signatures each generated by $N$ different signers on $N$ separate messages to be efficiently combined into a single compact signature. 
Following the framework introduced by Boneh \textit{et al.}~\cite{boneh2003aggregate}, 
we define syntax and security properties of \text{Aggregate Signature (\sf AS)} scheme.

\begin{definition}[Syntax]
Suppose $\lambda$ denotes the security parameter. An aggregate signature scheme is specified by a tuple of $5$-algorithms {\sf AS} = ({\sf AS.KeyGen}, {\sf AS.Sign}, {\sf AS.Verify}, {\sf AS.Aggregate}, {\sf AS.AggregateVerify}), defined as follows.

\begin{description}
\item[{\sf AS.KeyGen}$(1^\lambda)\rightarrow (pk_i,sk_i)$:] 
Given the security parameter $\lambda$, the Key Generation Centre (KGC) produces and distributes a corresponding public/secret key pair $(pk_i, sk_i)$ to the signer $i$, for each $i=1,\ldots, N$.

\item[{\sf AS.Sign}$(sk_i, \mu_i) \rightarrow v_i$:]
Using its secret key $sk_i$ and the message $\mu_i$ to be signed, the $i$-th signer generates the corresponding signature $v_i$ for $\mu_i$.
    
\item[{\sf AS.Verify}$(pk_i, \mu_i, v_i) \rightarrow 0/1$:] Provided the $i$-th signer’s public key $pk_i$ along with the message–signature pair $(\mu_i, v_i)$. If $(\mu_i, v_i)$ meets certain verification criteria, then output $1$ and accept; otherwise, output $0$ and reject. 
\item[{\sf AS.Aggregate}$(\{pk_i,\mu_i,v_i\}_{i=1}^N) \rightarrow v$:] 
For each $i = 1, \dots, N$, the aggregation algorithm takes the message–signature pairs $(\mu_i, v_i)$ along with their respective public verification keys $pk_i$, the aggregator produces the aggregate signature $v$.
    
\item[{\sf AS.AggregateVerify}$(v, \{pk_i,\mu_i\}_{i=1}^N)\rightarrow 0/1$:] 
Provided the aggregate signature $v$ together with the public keys $pk_i$ and messages $\mu_i$ for all $i = 1, \dots, N$. If the tuple $(v, \{pk_i,\mu_i\}_{i=1}^N)$ meets the specified verification conditions, produce $1$ to accept; otherwise, produce $0$ to reject.
\end{description}

\end{definition}

\noindent For the scheme \sf AS, it is essential to ensure correctness together existential unforgeability against chosen-message attack (EUF-CMA).\\

\noindent \underline{\textbf{Correctness.}}
For all \(\lambda, N \in \mathbb{N} \), for all key pairs \((pk_i, sk_i) \leftarrow \sf AS.KeyGen(1^\lambda)\), and for all messages \( \mu_i \in \mathcal{M} \), where $\mathcal{M}$ denotes the set of message space, and \(v_i \leftarrow \sf AS.Sign(sk_i,\mu_i)\) for all $i \in \{1, \ldots, N\}$:
$$\sf AS.AggregateVerify(\{pk_i\}_{i=1}^N, \{\mu_i\}_{i=1}^N, \sf AS.Aggregate(\{\sf AS.Sign(sk_i, \mu_i)\}_{i=1}^N)) = 1,$$ where $1$ denotes \texttt{True}.\\

\noindent \underline{\textbf{Security: Existential Unforgeability under Chosen Message Attack.}} Informally, the security of an aggregate signature scheme is defined by the infeasibility of an adversary successfully producing a forged aggregate signature within the parameters of a specified security game. In this context, \textit{existential forgery} refers to the ability of the adversary to create a valid aggregate signature over messages of his choosing, purportedly signed by a selected subset of users.\\
We rigorously define this notion using the aggregate chosen-key model for EUF-CMA. Within the framework of this model, an adversary $\mathcal{A}$ is given access to a single public key. The adversary's objective is to perform an existential forgery of an aggregate signature. $\mathcal{A}$ is granted the capability to select all public keys except for one \textit{challenge key}, and may also query a signing oracle associated with this challenge key. The adversary's advantage, denoted as $\text{Adv}^{\text{AggSig}}_{\mathcal{A}}$, is quantified by its probability of success in the following EUF-CMA game:

\begin{description}
\item[\textbf{Setup.}] The adversary $\mathcal{A}$ receives a public key $PK_1$, generated uniformly at random.
    
\item[\textbf{Queries.}] The adversary $\mathcal{A}$ may adaptively request signatures on messages of its choice under the challenge key $PK_1$.
    
\item[\textbf{Response.}] Eventually, $\mathcal{A}$ emits $k - 1$ extra public keys $PK_2, \ldots, PK_k$, where $k \leq N$, with $N$ being a predefined game parameter. These keys, together with $PK_1$, form the set of signers in the forged aggregate. $\mathcal{A}$ also outputs the corresponding messages $M_1, \ldots, M_k$ and an aggregate signature $v$, allegedly produced by these $k$ users.
\end{description}

\noindent The adversary wins the game if the aggregate signature $v$ is a valid existential forgery, that is, $v$ verifies correctly for the messages $M_1, \ldots, M_k$ under the public keys $PK_1, \ldots, PK_k$, and the forgery is nontrivial. That is, $\mathcal{A}$ did not make a signature request on $M_1$ under $PK_1$. The probability of $\mathcal{A}$'s success is determined based on the random choices made during the key generation and signing procedures of the scheme.

\subsection{Hierarchical Identity Based Signature}
A Hierarchical Identity-Based Signature (HIBS) lets users in an organization use their identity as their public key. They get private keys from their parent in the hierarchy and can sign messages. Anyone can verify signatures using just the signer's identity and the organization's root public key, removing the need for certificates.
Building on the framework proposed by Gentry \textit{et al.}~\cite{gentry2002hierarchical}, we describe the protocol below.

\begin{definition}[Syntax] Let $\lambda$ denote the security parameter. An HIBS scheme is defined by a tuple of five algorithms
{\sf HIBS}=(\sf HIBS.RootSetup, \sf HIBS.LowerLevelSetup, \sf HIBS.Extract, \sf HIBS.Sign, \sf HIBS.Verify)
\begin{description}
\item \textbf{\sf HIBS.RootSetup$(1^\lambda, 1^t)\rightarrow(pp, MSK)$:} The root private key generator (PKG) takes a security parameter $\lambda$ and the maximum depth $t$ as input and generates the system’s public parameters $pp$, and the master secret key $MSK$. The system’s public parameters include a description of the message space $\mathcal{M}$. $pp$ will be publicly available, while only the root PKG will know $MSK$.
\item \textbf{\sf HIBS.LowerLevelSetup$(pp,ID|_r) \rightarrow(SK_{ID|_r})$:} The algorithm takes the public parameter $pp$ and an identity $ID|_r$ of the child as input and outputs the secret key $SK_{ID|_r}$ for the given identity, which will be utilized to issue the secret signing key for the child identity $ID|_r$.
\item \textbf{\sf HIBS.Extract$(SK_{ID|_r}, ID|_k) \rightarrow SK_{ID|_k}$:} Given a parent’s secret key $SK_{ID|_r}$ for an identity $ID|_r = (ID_0, ID_1, \ldots, ID_r)$, and the identity $ID|_k$ for which secret key is to be extracted as input, the PPT algorithm produces a secret signing key $SK_{ID|_k}$ for the identity $ID|_k = (ID_0, \ldots, ID_r, \ldots, ID_k)$.
\item \textbf{\sf HIBS.Sign$(pp, SK_{ID|_k}, M) \rightarrow v$:} Using the public parameter $pp$, private key $SK_{ID|_k}$ and a message $M$ as input, the PPT algorithm generates a signature $v$ for the specified identity.
\item \textbf{\sf HIBS.Verify$(pp,ID|_k, M,v) \rightarrow (1/0)$:} Given the public parameter $pp$, a message–signature pair associated with an identity, If the signature verifies correctly, the algorithm returns $1$, if it fails verification, it returns $0$.
\end{description}
\end{definition}
\noindent For the HIBS scheme, it is crucial to guarantee correctness and Strong Unforgeability against adaptive identity and Chosen Message Attack (SU-aID-CMA) security.\\

\noindent \underline{\textbf{Correctness.}}
When $sk_{ID|_k}$ is the private key generated by the  HIBS.Extract algorithm for the identity $ID|_{k}$, then:
$$\forall M \in \mathcal{M}: \text{\sf HIBS.Verify}(pp,ID|_k, M, \text{\sf HIBS.Sign}(pp, SK_{ID|_k}, M)) = \text{valid}.$$\\
\noindent \underline{\textbf{Strong Unforgeability under adaptive identity and Chosen Message Attack.}}

\noindent The strongest security model for a Hierarchical Identity-Based Signature (HIBS) scheme is strong unforgeability against adaptive identity and chosen message attacks (SU-aID-CMA). This security notion is formalised through a security game involving a challenger $\mathcal{C}$ and an adversary $\mathcal{A}$, defined as follows:
\begin{description}
    \item \textbf{Setup:} To generate the system parameters $pp$, the Setup algorithm is executed by the challenger $\mathcal{C}$ and the adversary $\mathcal{A}$ is provided $pp$.

    \item \textbf{Extract Queries:} The adversary $\mathcal{A}$ can adaptively request private keys for any chosen identity $ID$. Upon receiving such a request, the challenger $\mathcal{C}$ executes the Extract algorithm to generate the corresponding secret key $S_{ID}$ and returns it to $\mathcal{A}$.

    \item \textbf{Sign Queries:} $\mathcal{A}$ may adaptively select an identity $ID$ along with a message $M$ and request a signature. The challenger $\mathcal{C}$ determine the signature $\nu = \text{\sf HIBS.Sign}(M, S_{ID})$ and returns it to $\mathcal{A}$.

    \item \textbf{Forgery:} Finally, the adversary $\mathcal{A}$ outputs a forged signature $(\nu^*, M^*, ID^*)$. The adversary is considered successful if
    $$\text{\sf HIBS.Verify}(\nu^*, M^*, ID^*) = 1$$
    and neither $ID^*$ nor any of its prefixes were queried during the Extract phase, and the pair $(M^*, ID^*)$ was not used in the Sign phase.
\end{description}

\noindent In summary, the adversary $\mathcal{A}$ wins the SU-aID-CMA game if it can produce a valid signature for a new message–identity pair without having obtained the corresponding private key or signature from the challenger.

\section{Hierarchical Identity-Based Signature with Designated Aggregator (HIBS-DA)} \label{HI}

Consider a hierarchical cryptographic infrastructure governed by a Root Private Key Generator (Root PKG) at level 0. The secret key and public key for the Root PKG is $MSK$ and $MPK$, respectively, where $MSK$ is defined as $SK_{ID|_{i0}}$, for each $i \in \{1, \ldots, N\}$. The system has a hierarchical depth $t$, comprising $t$ distinct levels beneath the root. These subordinate levels are sequentially indexed from $1$ (immediately below the root) to $t$ (the leaf level). Each hierarchical level $k$ ($1 \leq k \leq t$) contains a fixed population of $N$ entities, referred to as nodes or individuals.
Let $ID_{ik}$ denote the identity of $i$-th entity at hierarchical level $k$, where
$k \in \{1, 2, \dots, t\}$ denotes the hierarchical level and $i \in \{1, 2, \dots, N\}$ denotes the entity index within the level. A trapdoor delegation mechanism based on BasisDel (refer \ref{BD}) allows a parent in the hierarchy to securely derive a child’s private key without revealing its own trapdoor. This supports arbitrary depth up to $t+1$ while preserving security. The primary objective of this scheme is to perform cryptographic aggregation of digital signatures originating from a specific hierarchical level $k$, for some $1 \leq k \leq t$. This process aims to consolidate the $N$ distinct signatures generated by all members at level $k$ into a single, compact cryptographic attestation.
The aggregation mechanism proceeds as follows: Each entity $ID|_{ik}$, where $ID|_{ik}$ is defined as $ID|_{ik} = (ID_{i1}, \ldots, ID_{ik})$, for all $1 \leq k \leq N$ at hierarchical level $k$, generates signatures on their messages $M_{ik}$ using their respective private key $SK_{ID|_{ik}}$.
A designated aggregator collects all $N$ signatures and computes an aggregate signature $v_k$ for the hierarchical level $k$.
The aggregate signature $v_k$ can be verified against the public key corresponding to all $N$ members of level $k$. Upon completion, the Root PKG (or any authorized entity) possesses the aggregate signature $v_k$.

\begin{definition}[Syntax] Let $\lambda$ be the security parameter. An HIBS-DA scheme is formally defined by six algorithms
{\sf HIBS-DA}=(\sf HIBS-DA.Setup, \sf HIBS-DA.Derive, \sf HIBS-DA.Sign, \sf HIBS-DA.Verify, \sf HIBS-DA.Aggregate, \sf HIBS-DA.AggregateVerify).

\begin{description}
\item \textbf{\sf HIBS-DA.Setup($1^{\lambda},1^{t}) \rightarrow (pp, MPK, MSK)$:} Given a security parameter $\lambda$ and the maximum depth $t$, the setup algorithm generates the system public parameters $pp$, the master public key $MPK$ and the master secret key $MSK$ for the root PKG.
\item \textbf{\sf HIBS-DA.Derive($pp, SK_{ID|_{ir}}, ID|_{ik}) \rightarrow SK_{ID|_{ik}}$:} Given the system public parameter $pp$, private key $SK_{ID|_{ir}}$ corresponding to an identity  $ID|_{ir} = (ID_{i1}, \ldots, ID_{ir})$ of a parent and an identity $ID|_{ik} = (ID_{i1}, \ldots, ID_{ir}, \ldots, ID_{ik})$ of a child as input, the algorithm produces a secret key $SK_{ID|_{ik}}$ for the identity $ID|_{ik}$.
\item \textbf{\sf HIBS-DA.Sign$(ID|_{ik},M_{ik}, \mathsf{SK}_{ID|_{ik}}) \rightarrow v_{ik}$:} Given an identity $ID|_{ik}$, a private key $\mathsf{SK}_{ID|_{ik}}$ corresponding to the identity $ID|_{ik}$ and a message $M_{ik}$, the sign algorithm produces a signature $v_{ik}$ for the identity $ID|_{ik}$.
\item \textbf{\sf HIBS-DA.Verify($pp, {ID|_{ik}}, M_{ik}, v_{ik}) \rightarrow (1/0)$:} Given system public parameter, corresponding identity $ID|_{ik}$, message $M_{ik}$,  and signature $v_{ik}$, outputs $1$ if the signature is valid and $0$ otherwise.

\item \textbf{\sf HIBS-DA.Aggregate \big($\{M_{ik}\}_{i=1}^N ,\{ v_{ik}\}_{i=1}^N\big) \rightarrow v_k$:} An algorithm that, given message-signature pairs $(M_{ik}, v_{ik})$, for each $i= 1, \ldots, N$, the honest designated aggregator produces the aggregate signature $v_k$.

    \item \textbf{\sf HIBS-DA.AggregateVerify($\{M_{ik}\}_{i=1}^N, v_k, MPK) \rightarrow (1/0)$:} On input messages  $\{M_{ik}\}_{i=1}^N$, aggregate signature $v_k$, and the master public key $MPK$. It yields $1$, provided that the signature is valid. Otherwise, $0$.
\end{description}
\end{definition}
\noindent The HIBS-DA scheme must achieve correctness and Existential Unforgeability under Adaptive Identity and Chosen Message Attack security model.
\subsection{Correctness}
An HIBS-DA scheme is correct if for all security parameters $\lambda$, all hierarchical depths $t$, all identity tuples $ID|_{ik} = (ID_{i1}, \ldots, ID_{ik})$, and messages $M_{ik} \in \mathcal{M}$:
\begin{enumerate}
    \item \text{Individual Verification:} For any valid $(pp, MPK, MSK) \leftarrow \text{\sf HIBS-DA.Setup}(\lambda, t)$ and $SK_{ID|_{ik}} \leftarrow \text{\sf HIBS-DA.Derive}(pp, SK_{ID|_{ir}}, ID|_{ik})$:
    \[
    \Pr\left[\text{\sf HIBS-DA.Verify}\left(pp, ID|_{ik}, M_{ik}, v_{ik}\right) = 1\right] = 1
    \] 
    where $v_{ik} \leftarrow \text{\sf HIBS-DA.Sign}(ID|_{ik}, M_{ik}, SK_{ID|_{ik}})$
    
    \item \text{Aggregate Verification:} For any set of $N$ valid pairs $\{(M_{ik}, v_{ik})\}_{i=1}^N$:
    \[
    \Pr\left[\text{\sf HIBS-DA.AggregateVerify}\left(\{M_{ik}\}_{i=1}^N, v_k, MPK\right) = 1\right] = 1
    \]
    where $v_k \leftarrow \text{\sf HIBS-DA.Aggregate}\left(\{M_{ik}\}_{i=1}^N, \{v_{ik}\}_{i=1}^N\right)$.
\end{enumerate}
\subsection{Security Model}
The standard security requirement of the HIBS-DA scheme is Existential Unforgeability against Adaptive Identity and Chosen Message Attack (EUF-ID-CMA):
We define the security of our HIBS-DA scheme through a game played between a challenger $\mathcal{C}$ and a probabilistic polynomial-time (PPT) adversary $\mathcal{A}$, where the adversary aims to forge a valid aggregate signature for a specific hierarchical branch $k$ within the scheme.
\begin{definition}\textbf{Security Against Existential Unforgeable under Adaptive Identity and Chosen Message Attack (EUF-ID-CMA)}.
An HIBS-DA scheme is said to be 
EUF-ID-CMA secure
if for all PPT adversaries $\mathcal{A}$, the probability that $\mathcal{A}$ wins the above game is negligible in the security parameter $\lambda$:
$$\text{Adv}^{\text{HIBS-DA}}_{\mathcal{A}}(\lambda)  =\Pr[\mathcal{A} \text{ wins the EUF-ID-CMA game}] \leq \text{negl}(\lambda)$$    
\end{definition}
\noindent The security game is defined as follows:
\begin{description}
\item \textbf{Setup.} 
The challenger $\mathcal{C}$ runs the setup algorithm $(MPK, MSK) \leftarrow \textsf{Setup}(1^\lambda,1^t)$, where $\lambda$ is the security parameter and $t$ is the maximum hierarchy depth.
The challenger gives the master public key $MPK$ to the adversary $\mathcal{A}$ and keeps the master secret key $MSK$ private.

\item \textbf{Query Phase.} 
    The adversary $\mathcal{A}$ may adaptively make the following oracle queries:
\item[(i)] \textbf{Key Extraction Query:} 
$\mathcal{A}$ may request the private key for any hierarchical identity 
$ID|_{ik} = (ID_{i1}, \ldots, ID_{ik})$.
The challenger returns the corresponding secret key 
$SK_{ID|_{ik}}$.
$\mathcal{A}$ is not allowed to extract the secret key belonging to the target identity or any among its ancestors.

\item[(ii)] \textbf{Signing Query:} 
$\mathcal{A}$ may request a signature on any message $M_{ik}$ under a valid identity $ID|_{ik}$ except the target message.
The challenger responds with a valid signature
$v_{ik} \leftarrow \textsf{Sign}(SK_{ID|_{ik}}, M_{ik})$.

\item[(iii)] \textbf{Aggregate (Designated Aggregator) Query:} 
$\mathcal{A}$ may request the designated aggregator to combine multiple signatures 
$\{v_{ik}\}_{i=1}^{|S|}$ on messages $\{M_{ik}\}_{i=1}^{|S|}$ corresponding to identities $S = (ID|_{1k},\ldots, ID|_{{|S|}k})$.
The challenger responds with 
\[v_{k,S} \leftarrow \textsf{Aggregate}(\{M_{ik}, v_{ik}\}_{i=1}^{|S|}).\]
The adversary is restricted so that no identity appears with different messages in overlapping aggregate queries.

\item \textbf{Forgery.} 
    Finally, $\mathcal{A}$ outputs a candidate forgery $$(S^*, \{M_{ik}^*\}_{i \in S^*}, \{ID_{ik}^*\}_{i \in S^*}, v_{k,S^*}^*),$$ where $S^* = (ID^*|_{1k}, \ldots, ID^*|_{|S^*|k})$ is a set of hierarchical identities,
    $\{M_{ik}^*\}$ are messages, and $v_{k,S^*}^*$ is an aggregate signature, where $k$ is at most $N$.

\item \textbf{Winning Condition.}
    The adversary $\mathcal{A}$ wins the game if both of the following hold:
    \begin{itemize}
\item[(i)] The forged aggregate signature verifies correctly:
\[\textsf{Verify}(MPK, S^*, \{M_{ik}^*\}_{i=1}^{|S^*|}, v_{k,S^*}^*) = 1\].
\item[(ii)] There exists at least one identity $ID^*|_{jk} \in S^*$ such that
$\mathcal{A}$ did not query the key extraction oracle for $ID^*|_{jk}$ or any of its ancestors, and
$\mathcal{A}$ did not obtain a signature on message $M_{jk}^*$ for $ID^*|_{jk}$ in any prior signing or aggregation query.
\end{itemize}
\end{description}
\noindent In short, the HIBS-DA scheme is considered secure if no efficient adversary can generate a valid aggregate signature that includes at least one new identity–message pair for which it has never obtained the corresponding private key or legitimate signature.
                                           
\section{Our Construction} \label{CONS}
Consider a hierarchical cryptosystem of depth $t$, a Root PKG at level 0 governs $t$ lower levels. Each level $k$ (from $1$ to $t$) consists of $N$ entities, where each entity signs its message $M_{ik}$ with its private key $SK_{ID|_{ik}}$. A designated aggregator collects all $N$ signatures from level $k$ and computes a single aggregate signature $v_k$.\\ 
The HIBS-DA scheme is constructed below.\\

\noindent Let $q \geq 3$ be a prime and $m \geq 6n\log q$. Define $ID_0 = ID_{i0}$ is the root PKG’s identity and $ID|_{ik} = (ID_{i1}, \dots, ID_{ik})$ denotes the identity of $k$-th level entity of branch $i$ for each $i = {1,2, \ldots, N}$ . Let
$H : \{0,1\}^* \;\longrightarrow\; \mathbb{Z}_q^{m \times m}$ and $\sigma_{ik}, \sigma \geq \omega (\log\sqrt{m})$. The lattices $\mathbb{L}^{ik}$, for each $i = 1, 2, \ldots, N$, and for fixed $k \in \{1,2,\ldots,t\}$, are defined in (refer\ref{crt}).\\
Below, we detail our lattice-based HIBS-DA scheme.:\\

\noindent \fbox{\sf {HIBS-DA.Setup}$(1^{\lambda},1^{t}) \rightarrow ( MPK, MSK)$}\\

\noindent On input the security parameter $\lambda$ and the hierarchy depth $t$, a trusted third party does the following:\begin{enumerate}
    \item[(i)] Run $\textsf{TrapGen}(1^\lambda)$ to derive a matrix $\mathbf{A} \in \mathbb{Z}_q^{n \times m}$ and a short generating basis $\mathbf{T_A} \in \mathbb{Z}_q^{m \times m}$ for the lattice $\Lambda_q^\perp(\mathbf{A})$.
    \item[(ii)] Disclose the master public key $MPK$ = $\mathbf{A}$ and retain the master secret key $MSK$ = $\mathbf{T_A}$ confidential. 
    \item[(iii)] Consider ${SK}_{ID|_{i0}} = \mathbf{T_A}$ and $P_{ID|_{i0}} = \mathbf{A}$, for each $i \in \{1, \ldots, N\}$, are the secret key and public key of the root PKG, respectively. 
    \item[(iv)] Generate a basis $\mathcal{T}$ for the lattice $\mathbb{L}^{1k} \cap \mathbb{L}^{2k}\cap \ldots \cap \mathbb{L}^{Nk}$ by using Lemma \ref{homo} and give the basis $\mathcal{T}$ to the designated aggregator.
\end{enumerate}
\fbox{\sf {HIBS-DA.Derive}$(P_{ID|_{i0}}, SK_{ID|_{ir}}, \text{ID}|_{ik}) \rightarrow (SK_{ID|_{ik}})$}\\

\noindent On input the public key $P_{ID|_{i0}}$ of root PKG, a secret signing key $SK_{ID|_{ir}}$ corresponding to a “parent” identity $\text{ID}|_{ir} = (\text{ID}_{i1}, \ldots, \text{ID}_{ir})$, and a “child” identity $\text{ID}|_{ik} = (\text{ID}_{i1}, \ldots, \text{ID}_{ir}, \ldots, \text{ID}_{ik})$ where $k \leq t$:

\begin{enumerate}
    \item[(i)] Let $\mathbf{R_{\text{ID}|_{ir}}} = H(\text{ID}_{ir}) \cdots H(\text{ID}_{i2}) H(\text{ID}_{i1}) \in \mathbb{Z}_q^{m \times m}$ and\\
    \hspace*{1.5em}$P_{\text{ID}|_{ir}} = \mathbf{A} \mathbf{R^{-1}_{\text{ID}|_{ir}}} \in \mathbb{Z}_q^{n \times m}$. Then $SK_{\text{ID}|_{ir}}$ is a short reduced basis (refer \ref{bd}) for $\Lambda_q^\perp(P_{\text{ID}|_{ir}})$.
    \item[(ii)] Compute $\mathbf{R_{ik}} = H(\text{ID}_{ik}) \cdots H(\text{ID}_{i{(r+1)}}) \in \mathbb{Z}_q^{m \times m}$ and set $P_{\text{ID}|_{ik}} = P_{\text{ID}|_{ir}} \mathbf{R^{-1}_{ik}}$.
    
    \item[(iii)] Evaluate $\mathbf{S'_{ik}} \leftarrow \textsf{BasisDel}(P_{\text{ID}|_{ir}}, \mathbf{R_{ik}}, SK_{\text{ID}|_{ir}}, \sigma_{ik})$ to get a short, randomly generated basis for $\Lambda_q^\perp(P_{\text{ID}|_{ik}})$.
    
    \item[(iv)] Output the private key obtained through delegation $SK_{\text{ID}|_{ik}} = S'_{ik}$.

    \item[(v)] $SK_{\text{ID}|_{ik}} = S'_{ik}$ and $P_{\text{ID}|_{ik}} = P_{\text{ID}|_{ir}} \mathbf{R_{ik}^{-1}}$ are secret key and public key, respectively corresponding to the identity $ID|_{ik}$
    \end{enumerate}

\noindent \fbox{\sf {HIBS-DA.Sign}$(\text{ID}|_{ik}, SK_{ID|_{ik}}, M_{ik}) \rightarrow v_{ik} $}\\

\noindent Given a user $\text{ID}|_{ik}$ holding a secret signing key $\mathbf{SK}_{ID|_{ik}}$ along with a message $M_{ik} \in \{0,1\}^*$.
\begin{enumerate}
\item[(i)] Choose a string $b_{ik} \in \{0,1\}^n$, for $i = 1,2, \ldots, N$.
        \item[(ii)] Compute $v_{ik} \leftarrow \text{SamplePre}(P_{ID|_{ik}}, SK_{ID|_{ik}}, \sigma_{ik}, h(M_{ik}, b_{ik}, \text{ID}|_{ik}))$,\\
        where $h \colon \{0,1\}^* \to \mathbb{Z}_q^n$ is a hash function. Moreover, the resulting signature $v_{ik}$ satisfies $P_{ID|_{ik}} v_{ik}=h(M_{ik}, b_{ik}, \text{ID}|_{ik})$ and $\|v_{ik}\| \leq \sigma_{ik}\sqrt{m}$.

        \item[(iii)] Output the individual signature $v_{ik}$.
    \end{enumerate}
    
\noindent \fbox{\sf {HIBS-DA.Verify}$(P_{ID|_{ik}}, \text{ID}|_{ik}, v_{ik}, M_{ik}) \rightarrow \text{accept/reject}$}\\

\noindent For each $i = 1,2, \dots N$, on input the public key $P_{ID|_{ik}}$, identitity $\text{ID}|_{ik}$, where $k \in \{1,2, \dots t\}$, a signature $v_{ik}$ and a message $M_{ik}$, the verifier will accept the signature precisely when\\
$P_{ID|_{ik}} v_{ik} = h(M_{ik}, b_{ik}, \text{ID}|_{ik})$ and $\|v_{ik}\| \leq \sigma_{ik} \sqrt{m}$.\\

\noindent \fbox {\sf {HIBS-DA.Aggregate}$(\{M_{ik}\}_{i=1}^N,\{v_{ik}\}_{i=1}^N,{\mathcal{T}})\rightarrow v_k$}\\  

\noindent The honest designated aggregator performs the following operations to combine $N$ message-signature pairs $(M_{ik}, v_{ik})$ by using Chinese Remainder Theorem \ref{crth}.
 \begin{enumerate}
   \item[(i)] Computes $x \in \mathbb{Z}^m$ such that\\
\( x\equiv v_{1k} \pmod{\mathbb{L}^{1k}}\)\\
\( x \equiv v_{2k} \pmod{\mathbb{L}^{2k}} \)\\
\vdots\\
\( x \equiv v_{Nk} \pmod{\mathbb{L}^{Nk}} \)\\
\item[(ii)] This $x$ may not be short. 
\end{enumerate}
To get a short signature, the designated aggregator does the following: 
\begin{itemize}
\item[(a)] Output $v_k \leftarrow \text{SamplePre}(\mathbb{L}^{1k} \cap \mathbb{L}^{2k}\cap \ldots \cap \mathbb{L}^{Nk} ,{\mathcal{T}},x, \sigma)\in (\mathbb{L}^{1k} \cap \mathbb{L}^{2k}\cap \ldots \cap \mathbb{L}^{Nk}) +x$.
\item[(b)] The hierarchical identity-based signature with designated aggregator (HIBS-DA) is $v_k$.
\end{itemize}
 
\noindent \fbox{\sf {HIBS-DA.AggregateVerify}$(MPK, \text{ID}|_{ik}, v_k, \{M_{ik}\}_{i=1}^N)\rightarrow \text{accept/reject}$}\\

\noindent On input the master public key $\mathbf{A}$, identities $\text{ID}|_{ik}$, aggregate signature $v_k$ and messages $\{M_{ik}\}_{i=1}^N$, The verifier will accept the signature only when\\
     $\sum_{i=1}^N A[H(\text{ID}_{ik}), \dots, H(\text{ID}_{i1})]^{-1} v_k = \sum_{i=1}^N h(M_{ik}, b_{ik}, \text{ID}|_{ik})$ and $\|v_k\| \leq \sigma \sqrt{m}$.

\subsection{Correctness}
The correctness of SamplePre ensures that each $v_{ik}$ is short and satisfies the equation  
\begin{align*}
P_{\text{ID}|_{ik}} v_{ik} 
&= P_{\text{ID}|_{ir}} R_{ik}^{-1} v_{ik} \\[6pt]
&= A R_{\text{ID}|_{ir}}^{-1} R_{ik}^{-1} v_{ik} \\[6pt]
&= A \big[ H(\text{ID}_{ir}), \ldots, H(\text{ID}_{i1}) \big]^{-1}
      \big[ H(\text{ID}_{ik}), \ldots, H(\text{ID}_{i(r+1)}) \big]^{-1}v_{ik} \\[6pt]
&= A \big[ H(\text{ID}_{i1})^{-1}, \ldots, H(\text{ID}_{ir})^{-1}\big] \big[ H(\text{ID}_{i(r+1)})^{-1}, \ldots, H(\text{ID}_{ik})^{-1}\big]v_{ik} \\[6pt]
&= A \big[ H(\text{ID}_{ik}), \ldots, H(\text{ID}_{i1}) \big]^{-1} v_{ik}  \\[6pt]
&= h(M_{ik}, b_{ik}, \text{ID}_{ik}) .
\end{align*}

 As a result, the correctness of each individual signature is verified.

The solution derived from these equations\\
\begin{equation} \label{fg1}
\begin{cases}
 x\equiv v_{1k} \pmod{\mathbb{L}^{1k}}\\
  x \equiv v_{2k} \pmod{\mathbb{L}^{2k}}\\

\vdots\\
 x \equiv v_{Nk} \pmod{\mathbb{L}_{Nk}}\\
\end{cases}
\end{equation}
is a vector $x$. According to the CRT (refer \ref{crth}), $x$ corresponds uniquely to the tuple $( v_{1k}, \ldots, v_{Nk})$. However, $x$ might not be short. To shorten, we apply algorithm $\text{SamplePre}((\mathbb{L}^{1k} \cap \mathbb{L}^{2k}\cap \ldots \cap \mathbb{L}^{Nk}) ,\mathcal{T}, {x}, \sigma)$ to get $v_k$. This $v_k$ is short enough. i.e.,  $\|v_k\| \leq \sigma \sqrt{m}$, with $v_k$ being in the same coset as $x$.
It follows that $v_k$ satisfies equation \ref{fg1} and is of short length. i.e.,

\[
 v_k \bmod \mathbb{L}^{ik}= v_{ik} \bmod \mathbb{L}^{ik}.
\] 
We now observe that
\[\sum_{i=1}^N A[H(\text{ID}_{ik}), \ldots, H(\text{ID}_{i1})]^{-1} v_k = \sum_{i=1}^N h(M_{ik}, b_{ik}, \text{ID}|_{ik}).\]
Hence, $v_k$ authenticates all the messages $M_{ik}$, for $i = 1,2, \dots N$ and for some fixed $k \in \{1, \ldots, t\}.$\\

\subsection{Security Proof}
\begin{theorem}
If the SIS problem is computationally infeasible, then the proposed HIBS-DA scheme is existentially unforgeable under adaptive identity and chosen message attack (EUF-ID-CMA) secure. 
\end{theorem}
\begin{proof}
Suppose $\mathcal{A}$ is a probabilistic polynomial-time (PPT) adversary that aims to break the security of the HIBS-DA scheme, where $H$ and $h$ both are treated as random oracles, and let $Q_H$ and $q_h$ be the total queries performed by $\mathcal{A}$ respectively and $t$ be the maximum hierarchy depth, and suppose $\mathcal{A}$ forges the signature for the level $k$ (let), for $1 \leq k \leq t$, that is $\mathcal{A}$ forges the aggregate signature produced by the identities \{$ID|_{1k}, \ldots, ID|_{Nk}$\}, then there is a PPT algorithm $\mathcal{C}$ that solves the SIS problem.
We employ $\mathcal{A}$ to build an algorithm that solves the SIS problem.\\

\noindent \textbf{Setup.}
$\mathcal{C}$ sets up an environment that imitates attack conditions for $\mathcal{A}$ in the following manner.

\begin{enumerate}
    \item Pick $N \cdot t$ integers $Q^*_{i1}, \ldots, Q^*_{it} \in [Q_H]$, where $i \in \{1, \ldots, N\}$ uniformly at random and $Q_H$ specifies the maximum number of queries to $H$ that $\mathcal{A}$ can submit for each $i$.
    \item Sample $N \cdot t$ matrices $R^*_{i1}, \ldots, R^*_{it} \sim D_{m \times m}$ uniformly at random by executing $R^*_{ij} \leftarrow \mathsf{SampleR}(1^m)$ for $j = 1, \ldots, t$ and $i = 1, \ldots, N$.
    \item Construct a random matrix $\mathbf{A}_{0} \in \mathbb{Z}_q^{n \times m}$.
    \item Pick $w \in [t]$ uniformly at random and assign 
$\mathbf{A}_i \leftarrow \mathbf{A}_{0}R^*_{iw} \cdots R^*_{i1}, \quad i = 1, \ldots, N.$
Because $\mathbf{A}_0$ is uniform in $\mathbb{Z}_q^{n \times m}$ and all $R^*_{ij}$ are invertible modulo $q$, each $\mathbf{A}_i$ is uniformly distributed over $\mathbb{Z}_q^{n \times m}$.
\item Issue parameter $pp = \mathbf{A}_i$ publicly.
\end{enumerate}
\textbf{Random-oracle hash queries $(H)$.}
The adversary $\mathcal{A}$ is permitted to, at any time, adaptively query the random oracle $H$ on any identity
$ID|_{ik} = (ID_{i1}, \ldots, ID_{ik})$ of its choosing.
The challenger $\mathcal{C}$ handles the $Q_H$-th distinct query as described below.
For simplicity, we assume that all queries are unique; if $\mathcal{A}$ repeats a query, the simulator returns the previously stored output and does not increment the counter $ Q_H$.\\

\noindent Let $i = |(ID|_{ik})|$ denotes the identity's $ID|_{ik}$ depths. When such a query corresponds to $Q_{ik}^*$ (meaning $Q_H = Q_{ik}^*$), we proceed to assign $H(ID|_{ik}) \leftarrow R_{ik}^*$ and output $H(ID|_{ik})$.\\
In all other cases, when $Q_H \neq Q_{ik}^*$:
\begin{enumerate}
\item Evaluate $\mathbf{A}_{ik} = \mathbf{A_i} \cdot (R^*_{(i-1)k} \cdots R^*_{2k} R^*_{1k})^{-1} \in \mathbb{Z}_q^{n \times m}$ (where $\mathbf{A}_{1k} = \mathbf{A}_1$).
\item Execute $\mathsf{SampleRwithBasis}(\mathbf{A}_{ik})$ to generate $R_k \sim D_{m \times m}$ uniformly at random and a concise basis $\mathbf{T}_{\mathbf{B}_k}$ for $\mathbf{B}_k = \mathbf{A}_{ik} R_k^{-1} \bmod q$.
\item Store the $5$-tuple $(i, ID|_{ik}, R_k, \mathbf{B}_k, \mathbf{T}_{\mathbf{B}_{k}})$ for later employ, and output $H(ID|_{ik}) \leftarrow R_k$.
\end{enumerate}
\textbf{Secret key queries.} The adversary $\mathcal{A}$ adaptively make interactive queries that extract keys for arbitrary identities $ID|_{ik}$, of its choice. $\mathcal{C}$ responds a query on $ID|_{mk} = (ID_{1k}, ID_{2k}, \ldots, ID_{mk})$ of length $|(ID|_{mk})| = m \in [t]$ as detailed below.

\begin{enumerate}
    \item Let us suppose $j \in [m]$ denote the earliest level where $H(ID|_{jk}) \neq R_{jk}^*$. If, in the unexpected case, $H(ID|_{jk}) = R_{jk}^*$ for every $j = 1, \ldots, m$, 
the simulator terminates, at which point the process does not succeed.
    
    \item Recover the previously stored tuple $(j, ID|_{jk}, R_k, \mathbf{B}_k, \mathbf{T}_{\mathbf{B}_k})$ from the recorded query history of the hash oracle $H$.
This tuple arises from answering a query to $H(ID|_{jk})$.
(Without loss of generality, we may suppose that any extraction query on $ID|_{mk}$
is introduced by queries to the hash oracle across all prefixes of $ID|_{mk}$.)
As a consequence of the construction,
\[
\mathbf{B}_k = \mathbf{A}_j \cdot (R_{1k}^*)^{-1} \cdots (R_{(j-1)k}^*)^{-1} \cdot H(ID|_{jk})^{-1} \mod q,
\]
and $\mathbf{T}_{\mathbf{B}_k}$ constitutes a concise basis of $\Lambda_q^\perp(\mathbf{B}_k)$.
    
Observe that $\mathbf{B}_k$ is precisely the signature matrix $P_{ID|_{jk}}$ (as specified by the signing algorithm) for identity $ID|_{jk} = (ID_{1k}, ID_{2k}, \ldots, ID_{jk})$ of the ancestor and therefore $\mathbf{T}_{\mathbf{B}_k}$ is a trapdoor for $\Lambda_q^\perp(P_{ID|_{jk}})$.
    
\item Execute BasiDel$(P_{ID|_{jk}}, H(ID|_{mk}) \ldots H(ID|_{(j+1)k}), \mathbf{T}_{\mathbf{B}_k}, \sigma_{jk})$ to produce a secret signing key for $ID|_{jk}$ and delivers the resulting secret key to $\mathcal{A}.$
\end{enumerate}
\textbf{Random-oracle hash queries($h$) and sign queries.} $\mathcal{C}$ generates a list $L$ to record the responses to the $h$ queries.
$\mathcal{C}$ runs $\mathcal{A}$ on public key $P_{ID|_{ik}}$, for $i \in \{1, \ldots, N\}$ and $\mathcal{A}$ declares to $\mathcal{C}$ the identity $ID_{ik}^*$ chosen for the challenge along with a message $M_{ik}^*$ to be signed and provides a simulation of the random oracle $h$ together with signing oracle as follows.
Since, $k= |(ID_{ik}^*)|$. Recall that $A_i = A_{0} R^*_{iw} \cdots R^*_{i1}$. If $w \neq k$ then the simulator terminates and does not succeed.
Next, assume $w=k$ and $ID_{ik}^*$ is satisfying $H(ID_{ik}^*) = R_{ik}^*$ for all $i \in \{1, \ldots, k\}$. It follows directly from the definition that
$P_{ID|_{ik}^*} = \mathbf{A_i} (R_{i1}^*)^{-1} \ldots (R_{ik}^*)^{-1} = \mathbf{A_0} \in \mathbb{Z}_q^{n \times m}.$
Assume, without loss of generality that $\mathcal{A}$ queries $h$ for each message $M_{ik}$ prior to issuing a signing query on $M_{ik}$ for fixed $k$.
\begin{itemize}
\item For each query made to $h$ on a different $M_{ik} \in \{0, 1\}^*$, $\mathcal{C}$ chooses $a_{ik} \in \{0,1\}^n$ uniformly at random, runs $v_{ik} \leftarrow \text{SampleDom}(1^n)$, and stores $(M_{ik}, a_{ik}, ID|_{ik}, h, v_{ik})$ in the list $L$, and returns $P_{ID|_{ik}} v_{ik} = h(M_{ik}, a_{ik}, ID|_{ik})$ to $\mathcal{A}$. (If $h$ already been queried on $M_{ik}$, $\mathcal{C}$ searches for $(M_{ik}, v_{ik})$ and supplies $P_{ID|_{ik}} v_{ik}$.)
    
\item Whenever $\mathcal{A}$ issues a signing query on the pair $(M_{ik},ID|_{ik})$, $\mathcal{C}$ looks up the list in its local memory and outputs $v_{ik}$ as the resulting signature.
\end{itemize}
Let us suppose, without loss of generality that prior to producing its attempted forgery $(M_{ik}^{*}, ID|_{ik}^*, v_{ik})$, $\mathcal{A}$ queries $h$ on $(M_{ik}^{*}, ID|_{ik}^*)$ and produces the signature  $v_{ik}^{*}$ for the message-identity pair  $(M_{ik}^{*}, ID|_{ik}^*)$.\\

\noindent \textbf{Aggregate queries.} The adversary $\mathcal{A}$ is allowed to request aggregate signatures on any set of tuples 
$\{(M_{ik}, ID_{ik}, v_{ik})\}_{i=1}^{N}$ of its choice except the list containing challenged message-identity pair.
Upon receiving such a request, the challenger $\mathcal{C}$ computes a compact aggregate 
signature by employing the Chinese Remainder Theorem (CRT) and the sampling algorithm 
$\textsf{sampleDom}(1^n)$. The resulting HIBS-DA signature $v_k$ satisfies
\[
\sum_{i=1}^{N} P_{ID|_{ik}} \cdot v_k 
  = \sum_{i=1}^{N} h(M_{ik}, a_{ik}, ID|_{ik}).
\]

\noindent \textbf{Forgery.}
After obtaining sufficient information through oracle queries, 
$\mathcal{A}$ outputs a forged aggregate signature $v_k^*$. 
In order for the forgery to be valid, the adversary must include at least one 
identity--message pair $(ID_{ik}^*, M_{ik}^*)$ that was not queried to the signing oracle, 
although it may have queried the hash oracle on that pair.

Upon receiving the forged signature $v_k$ corresponding to the tuples 
$\{(M_{ik}, ID_{ik}, v_{ik})\}_{i=1}^{N}$, 
the challenger $\mathcal{C}$ inspects its record of message-identity pairs 
and identifies the corresponding entry $(M_{ik}^*, ID_{ik}^*, v_{ik}^*)$. 
It then compares the legitimate HIBS-DA signature $v_k^*$ with the forged one, obtaining
\[
A_0 \cdot (v_k - v_k^*) = 0,
\]
which constitutes a valid solution to an SIS instance.\\
We need to only show that $v_k \neq v_k^*$.
Applying the preimage min-entropy property (refer \ref{meh}) of the employed family of hash functions, the conditional min-entropy of $v_k$, given 
$\mathbf{A}_i[H(ID_{ik}) \ldots H(ID_{i1})]^{-1} \cdot v_k$ (and the remaining view obtained by 
$\mathcal{A}$, which is not dependent on $v_k$) is $\omega(\log n)$. 
Consequently, the probability that $v_k = v_k^*$ is negligible, 
and thus the forgery succeeds only with negligible probability.
\end{proof}

\section{Efficiency Analysis} \label{eff ana}
This section presents the efficiency analysis of our proposed HIBS-DA scheme.
We describe the detailed round complexity, the communication complexity, and the computation complexity of our proposed HIBS-DA sheme below.\\

\noindent \textbf{Round Complexity}: Our proposed HIBS-DA scheme is \textbf{non-interactive}. The signing process is executed locally by each signer individually by using their corresponding secret key. Subsequently, aggregation requires only a single, one-way transmission of individual signatures from the signers to the designated aggregator, who then performs the aggregation as a local computation. Since an interactive protocol is formally defined as one requiring multiple rounds of communication back and forth between parties, the core operations of our scheme, signing and aggregation, exhibit a round complexity of zero.\\

\noindent \textbf{Communication Complexity}: The communication complexity of our proposed HIBS-DA scheme is analyzed in terms of the total number of bits exchanged among the involved entities, namely the root PKG, subordinate PKGs, signers, the designated aggregator, and the verifier. 
Each phase of the scheme involves the transmission of lattice-based matrices or vectors whose sizes depend on the lattice parameters $m$, $n$, and the modulus $q$. 
Specifically, each element in $\mathbb{Z}_q$ can be represented using $\lceil \log_2 q \rceil$ bits, which allows us to express the communication cost in terms of bit lengths. 
Table~\ref{tab:comm_bits} summarizes the communication requirements of each phase in bits.
\begin{table}[H]
\centering
\caption{Communication complexity of our proposed HIBS-DA scheme}
\renewcommand{\arraystretch}{1.2}
\begin{tabular}{|l|l|l|l|}
\hline
\textbf{Phase} & \textbf{Communication Parties} & \textbf{Data Transmitted} & \textbf{Bit Complexity} \\ \hline
\textbf{Setup} & PKG $\rightarrow$ Public & Publish $\mathbf{A} \in \mathbb{Z}_q^{n \times m}$ & $n m \lceil \log_2 q \rceil$ bits (once) \\ \hline
\textbf{Key Derivation} & Parent $\rightarrow$ Child & Send $SK_{ID|_{ir}} \in \mathbb{Z}_q^{m \times m}$ & $m^2 \lceil \log_2 q \rceil$ bits per delegation \\ \hline
\textbf{Signing} & (Local operation) & None & $0$ bits \\ \hline
\textbf{Aggregation} & $N$ Signers $\rightarrow$ Aggregator & Send $N$ signatures $v_{ik} \in \mathbb{Z}_q^m$ & $N m \lceil \log_2 q \rceil$ bits \\ \hline
\textbf{Verification} & Aggregator $\rightarrow$ Verifier & Send $v_k \in \mathbb{Z}_q^m$ & $m \lceil  \log_2 q \rceil$ bits \\ \hline
\end{tabular}
\label{tab:comm_bits}
\end{table}
\noindent The total communication cost of our HIBS-DA scheme is obtained by summing the communication requirements of all components.
Since the hierarchical depth is $(t+1)$, there are $t$ key delegation steps. 
Hence, the total communication complexity in bits is:
\begin{equation}
\begin{aligned}
C_{\text{total}} 
&= (n m \lceil \log_2 q \rceil) + (t \cdot m^2 \lceil \log_2 q \rceil) + (N m \lceil \log_2 q \rceil) + (m \lceil \log_2 q \rceil) \\
&= (m + t m^2 + n m + N m) \lceil \log_2 q \rceil \\
&\approx O\big((t m^2 + N m)\lceil \log_2 q \rceil\big) \text{ bits.}
\end{aligned}
\end{equation}

\noindent \textbf{Computation complexity}: 
The computation cost of each algorithm in our proposed lattice-based HIBS-DA scheme is summarized in Table~\ref{tab:complexity}. 
The analysis is carried out relative to the lattice parameters \( n \), \( m \), modulus \( q \), hierarchy depth \( k \), and the number of aggregated users \( N \). 
Here, \( \tilde{O}(\cdot) \) hides polylogarithmic factors in \( n \) and \( q \). 

During the Setup phase, the TrapGen algorithm dominates the cost, resulting in complexity \( \tilde{O}(n^2) \). 
The Derive phase involves basis delegation and matrix multiplications at each hierarchical level, with an overall complexity of \( \tilde{O}(k \cdot m^3) \). 
Both the Sign and Verify phases rely primarily on Gaussian sampling and matrix--vector operations, yielding \( \tilde{O}(m^2) \) and \( O(m^2) \) complexities, respectively. 
The Aggregate phase introduces an additional Chinese Remainder Theorem (CRT) computation \( O(Nm) \) and short vector sampling \( \tilde{O}(m^2) \). 
Finally, the AggregateVerify phase performs \( N \) parallel matrix--vector multiplications, resulting in \( O(Nm^2) \) complexity.\\
\noindent\textit{Note:} \( \tilde{O}(\cdot) \) hides polylogarithmic factors in \( n \) and \( q \).\\
\begin{table}[h!]
\centering 
\caption{Computation Complexity of our proposed HIBS-DA scheme}
\begin{tabular}{|l|l|l|}
\hline
\textbf{Algorithm} & \textbf{Main Cost} & \textbf{Asymptotic Complexity} \\ \hline
\textbf{Setup} & Trapdoor generation & $\tilde{O}(n^2)$ \\ \hline
\textbf{Derive} & Basis Delegation & $\tilde{O}(k \cdot m^3)$ \\ \hline
\textbf{Sign} & Gaussian Sampling & $\tilde{O}(m^2)$ \\ \hline
\textbf{Verify} & Matrix-Vector Multiplication & $O(m^2)$ \\ \hline
\textbf{Aggregate} & CRT + intersection lattice sampling & $O(N m) + \tilde{O}(m^2)$ \\ \hline
\textbf{AggregateVerify} & \(N\) Matrix-Vector Multiplications & $O(N \cdot m^2)$ \\ \hline
\end{tabular}
\label{tab:complexity}
\end{table}

\noindent In this analysis, we have discussed the round complexity, the computation complexity, and
 the communication complexity in terms of bits where the communication cost depends on the parameters $n,m,q$.

\section{Conclusion}
In this paper, we provided the {\em{first}} hierarchical identity-based signature with designated aggregator protocol that effectively addresses the challenges of secure, efficient, and scalable authentication in multi-level organizational settings. Our analysis confirms that the proposed scheme ensures correctness, unforgeability, and consistency, while maintaining scalability and practicality within the hierarchy, making HIBS-DA a strong candidate for deployment in modern decentralized systems that require structured authentication with minimal overhead. Our scheme is non-interactive with zero round complexity. Moreover, the communication complexity is
$O\!\left((t m^{2} + N m)\,\lceil \log_{2} q \rceil\right)$ bits, and the scheme achieves computation complexities of $\tilde{O}(n^{2})$ for Setup, $\tilde{O}(k m^{3})$ for Derive, $\tilde{O}(m^{2})$ for Sign, $O(m^{2})$ for Verify, $O(N m) + \tilde{O}(m^{2})$ for Aggregate, and $O(N m^{2})$ for AggregateVerify.

\bibliography{Ref2.bib}
\bibliographystyle{unsrt}

\section{Appendix} 
\subsection{A Specific Case of the Chinese Remainder Theorem}\label{crt}
In our proposed construction, we employ Theorem \ref{crth} to decompose the underlying ring structure and facilitate efficient computation. Specifically, we consider the commutative ring 
\[
R = \mathbb{Z}^m \text{ and for each } i \in \{1, \dots, N\} \text{ and for some fixed } k, \; S_i = \mathbb{L}^{ik}.
\]

\noindent In particular, let us suppose $\mathbf{e}_i = (y_1, \ldots, y_i, \ldots, y_m)$, where $y_j = 1$ for $i = j$ and $y_j = 0$ for $i \neq j$ with $i= 1, \ldots, m$. Suppose $P_{ID|_{ik}} \in \mathbb{Z}_q^{n \times m}$ for $i = 1, \ldots, N$ and for some fixed $k \in \{1, \ldots, t\}$ (According to the construction \ref{CONS}).\\     
If $N<m$, consider 
\begin{align*}
    &\mathbb{L}^{1k}=\mathbb{L}_q^{\perp}(P_{ID|_{1k}})+\langle \mathbf{e}_2, \mathbf{e}_3, \ldots, \mathbf{e}_m\rangle,\nonumber\\
    & \mathbb{L}^{2k}=\mathbb{L}_q^{\perp}(P_{ID|_{2k}})+\langle \mathbf{e}_1, \mathbf{e}_3, \ldots, \mathbf{e}_m\rangle,\nonumber\\
    &\vdots\\
    & \mathbb{L}^{(N-1)k}=\mathbb{L}_q^{\perp}(P_{ID|_{(N-1)k}})+\langle \mathbf{e}_1,  \ldots, \mathbf{e}_{N-2}, \mathbf{e}_N, \ldots, \mathbf{e}_m\rangle,\nonumber\\ 
    & \mathbb{L}^{Nk}=\mathbb{L}_q^{\perp}(P_{ID|_{Nk}})+\langle \mathbf{e}_1, \ldots, \mathbf{e}_{N-1}, \mathbf{e}_{N+1}, \ldots, \mathbf{e}_m \rangle.\nonumber
\end{align*}
If $N \geq m$, consider
\begin{align*}
    &\mathbb{L}^{1k}=\mathbb{L}_q^{\perp}(P_{ID|_{1k}})+\langle \mathbf{e}_2, \mathbf{e}_3, \ldots, \mathbf{e}_m\rangle, \nonumber\\
    &\mathbb{L}^{2k}=\mathbb{L}_q^{\perp}(P_{ID|_{2k}})+\langle \mathbf{e}_1, \mathbf{e}_3, \ldots, \mathbf{e}_m \rangle,\nonumber\\
    &\vdots\\
    &\mathbb{L}^{mk}=\mathbb{L}_q^{\perp}(P_{ID|_{mk}})+\langle \mathbf{e}_1, \mathbf{e}_2, \ldots, \mathbf{e}_{m-1}\rangle,\nonumber\\ &\mathbb{L}^{(m+1)k}=\mathbb{L}_q^{\perp}(P_{ID|_{(m+1)k}}) + \langle \mathbf{e}_1, \mathbf{e}_2, \ldots, \mathbf{e}_m \rangle,\nonumber\\
    &\vdots \\
    &\mathbb{L}^{Nk}=\mathbb{L}_q^{\perp}(P_{ID|_{Nk}}) + \langle \mathbf{e}_1, \mathbf{e}_2, \ldots, \mathbf{e}_m  \rangle.\nonumber
\end{align*}
The lattice $\mathbb{L}^{ik}$ is well defined for $i= 1, \ldots, N$. 
\begin{align*}
    \mathbb{L}^{ik}=&\mathbb{L}_q^{\perp}(P_{ID|_{ik}})+\langle \mathbf{e}_1, \ldots,  \mathbf{e}_{i-1}, \mathbf{e}_{i+1}, \ldots, \mathbf{e}_m \rangle\\
    =&\left\{ \mathbf{x}=(x_1, \ldots, x_m) \in \mathbb{Z}^m : P_{ID|_{ik}} \mathbf{x}=\mathbf{0} \bmod q \right\}\\
    &\hspace{20pt}+\left\{ a_{1k} e_1 + a_{2k} e_2 + \ldots + a_{(i-1)k} e_{i-1} + a_{(i+1)k} e_{i+1} + \ldots a_{mk} e_m \right\} \\
    =&\left\{ (a_{1k} x_1, a_{2k} x_2, \ldots, a_{(i-1)k} x_{i-1}, x_i, a_{(i+1)k} x_{i+1}, \ldots, a_{mk} x_m) \in \mathbb{Z}^m : x_i, a_{ik} \in \mathbb{Z}, \forall i \in \{1,2, \ldots m\} \right\}.
\end{align*}
So, we obtain $N$ lattices $\mathbb{L}^1, \mathbb{L}^2, \ldots, \mathbb{L}^N$ such that for $i \neq j$,
\begin{align*}
    \mathbb{L}^i + \mathbb{L}^j
    =&\bigl(\mathbb{L}_q^{\perp}(P_{ID|_{ik}}) + \langle \mathbf{e}_1, \ldots, \mathbf{e}_{i-1}, \mathbf{e}_{i+1}, \ldots, \mathbf{e}_m \rangle\bigr)\\
    &+ \bigl(\mathbb{L}_q^{\perp}(P_{ID|_{jk}}) + \langle \mathbf{e}_1, \ldots, \mathbf{e}_{j-1}, \mathbf{e}_{j+1}, \ldots, \mathbf{e}_m \rangle\bigr)\\
    =&\mathbb{Z}^m.
\end{align*}
\noindent Since, $R = \mathbb{Z}^m$ be a ring. For each $i \in \{1, \dots, N\}$ and for some fixed $k$, $S_i = \mathbb{L}^{ik}$ be ideals of $R$ and we have already established that $\mathbb{L}^i + \mathbb{L}^j = \mathbb{Z}^m$ for $i \neq j$. Then by CRT. (refer~\ref{crth}), 
$$\frac{\mathbb{Z}^m}{\mathbb{L}^{1k} \cap \mathbb{L}^{2k} \cap \cdots \cap \mathbb{L}^{Nk}} \cong \frac{\mathbb{Z}^m}{\mathbb{L}^{1k}} \times \frac{\mathbb{Z}^m}{\mathbb{L}^{2k}} \times \cdots \times \frac{\mathbb{Z}^m}{\mathbb{L}^{Nk}} 
.$$ Therefore, for some fixed $k$, $\exists$ a ring isomorphism 
\[\frac{\mathbb{Z}^m}{\bigcap_{i=1}^N \mathbb{L}^{ik}} \;\cong\; \prod_{i=1}^{N} \frac{\mathbb{Z}^m}{\mathbb{L}^{ik}}, \] 
given by
\[t \;\mapsto\;(t\bmod\mathbb{L}^{1k}, t\bmod\mathbb{L}^{2k}, \ldots, t\bmod\mathbb{L}^{Nk}).\]

\subsection{Integration of the CRT into our proposed HIBS-DA scheme} \label{INT}
In this subsection, we show how the CRT \ref{crt} is applied to our proposed HIBS-DA scheme.\\
From Section \ref{CONS}, in Sign algorithm, we observe that the individual signature $v_{ik}$, for each $i \in \{1, \ldots, N\}$ and fixed $k \in \{1, \ldots, t\}$, belongs to the lattice $\mathbb{L}_q^{h}(P_{ID|_{ik}})$, where $h = h(M_{ik}, b_{ik}, \text{ID}|_{ik})$ and $\mathbb{L}_q^{h}(P_{ID|_{ik}}) = \left\{ \mathbf{t} \in \mathbb{Z}^m : P_{ID|_{ik}}\mathbf{t} = h(M_{ik}, b_{ik}, \text{ID}|_{ik}) \pmod q \right\}$\\
\noindent Since
\begin{align*}
    v_{ik} &\in \mathbb{L}_q^{h}(P_{ID|_{ik}}) 
            = \mathbb{L}_q^{\perp}(P_{ID|_{ik}}) + \mathbf{t_1}, 
            \quad \text{for some } \mathbf{t_1} \text{ satisfying } 
            P_{ID|_{ik}} \mathbf{t_1} = h \pmod{q}. \\
    v_{ik} &\in \mathbb{L}_q^{\perp}(P_{ID|_{ik}}) + \mathbf{t_1} 
            \implies
            v_{ik} \in \big(\mathbb{L}_q^{\perp}(P_{ID|_{ik}}) + \mathbf{t_1}\big) 
            + \langle \mathbf{e}_1, \ldots, \mathbf{e}_{i-1}, \mathbf{e}_{i+1}, \ldots, e_{m} \rangle \\
            &= \big(\mathbb{L}_q^{\perp}(P_{ID|_{ik}}) 
            + \langle \mathbf{e}_1, \ldots, \mathbf{e}_{i-1}, \mathbf{e}_{i+1}, \ldots, e_{m} \rangle \big) + \mathbf{t_1}, \mathbf{t_1} \in \mathbb{Z}^m\\
            &\implies
            v_{ik} \in \frac{\mathbb{Z}^{m}}{\mathbb{L}^{ik}}. 
\end{align*}
Therefore, 
\[
(v_{1k}, v_{2k}, \ldots, v_{Nk}) \in 
\frac{\mathbb{Z}^{m}}{\mathbb{L}^{1k}} 
\times 
\frac{\mathbb{Z}^{m}}{\mathbb{L}^{2k}} 
\times \cdots \times 
\frac{\mathbb{Z}^{m}}{\mathbb{L}^{Nk}},
\]
by CRT \ref{crth}, there exists an element 
\[
x \in \frac{\mathbb{Z}^{m}}
{\mathbb{L}^{1k} \cap \mathbb{L}^{2k} \cap \cdots \cap \mathbb{L}^{Nk}}
\]
such that
\begin{align*}
    x &\equiv v_{1k} \pmod{\mathbb{L}^{1k}}, \\
    x &\equiv v_{2k} \pmod{\mathbb{L}^{2k}}, \\
    &\vdots \\
    x &\equiv v_{Nk} \pmod{\mathbb{L}^{Nk}}
\end{align*}
defined in \ref{crt}.

\end{document}